\documentclass{article}
\usepackage{spconf,amsmath,graphicx}
\usepackage{epstopdf}
\usepackage{amsthm}
\usepackage{fancyhdr}
\newcommand{\bm}[1]{\mbox{\pmb{$#1$}}}

\newtheorem {Lemma}{Lemma}

\title{\Large Multi-User Communication in Difficult Interference}
\name{Dushyantha A. Basnayaka, and Tharmalingam Ratnarajah}
\address{Institute for Digital Communication, School of Engineering, \\ University of Edinburgh, Mayfield Road, Edinburgh EH9 3JL\\E-mail: d.basnayaka@ed.ac.uk}
\begin{document}
%
%
\maketitle
\begin{abstract}
The co-channel interference (CCI) is one of the major impairments in wireless communication. CCI typically reduces the reliability of wireless communication links, but the ``difficult'' CCI which is no more or less strong to the desired signal destroys wireless links despite having myriad of CCI mitigation methods. It is shown in this paper that $M$-QAM (Quadrature Amplitude Modulation) or similar modulation schemes which modulate information both in in-phase and quadrature-phase are particularly vulnerable to difficult CCI. Despite well-known shortcomings, it is shown in this paper that $M$-PAM or similar schemes that use a single dimension for modulation provides an important mean for difficult CCI mitigation. 
\end{abstract}
\section{Motivation}\label{sec:motivation}
Consider a typical multi-user communication scenario \cite{Paulraj03}, where many sensor nodes communicate with an entry-level access point (AP). One of the simplest methods to allow all users (or nodes) to communicate reliably is to use time division multiple access (TDMA) with only one node being given the exclusive access to the channel in a given time interval \cite{Proakis91}. However, for the sake of higher system spectral efficiency, one can pair users, and allow co-channel communication forming a multiple-access channel (MAC). A communication scenario with such two co-channel users (both users and AP with a single antenna each) gives rise to the complex scalar channel with interference \cite{Zandar92}:
\begin{align}
y &= h_1 s_1 + h_2 s_2 + n, \nonumber
\end{align}
where $s_1$ and $s_2$ are source symbols drawn from digital constellation (say BPSK), and $h_1$ and $h_2$ are scaler complex fading coefficients from node 1 and node 2 to AP respectively. $n$ is additive noise. It is now well-known that if the relative channel strength of users, which is typically statistically measured by the ratio between $\mathcal{E}\left\{\left|h_1\right|^2\right\}$ and $\mathcal{E}\left\{\left|h_2\right|^2\right\}$ is not unity, the inference estimation and subtraction (i.e., SIC often used at AP in technologies like NOMA \cite{SaKiBe13}) can be used. However, in difficult co-channel interference (CCI), which is no more or less stronger to the desired signal, AP can not use simple single user detection techniques at all (they fail hopelessly). Space domain CCI mitigation techniques cannot be used owing to the lack of multi-antennas at the AP \cite{Paulraj03}. The last resort is to use joint maximum likelihood (ML) detection. For instance, data of user 1 can be detected as \cite{Stuber02}:
\begin{align}\label{eq:ML:A}
\hat{s}_1=\min_{s_1,s_2 \in \mathcal{Q}} \left|y - h_1 s_1 - h_2 s_2\right|,
\end{align}
where $\mathcal{Q}$ is the constellation of $s_1$ and $s_2$. In this paper, a simple feedback signal processing scheme (in its simplest form that exploits single-phase modulation schemes like BPSK or $M$-PAM ) is proposed and analyzed to sustain co-channel communication in difficult interference. It has been shown that the proposed scheme not only uses simple single user detection, but also outperforms even the ML scheme in \eqref{eq:ML:A}, which is conventionally regarded as a fundamental limit.
\section{Main System Model}\label{sec:main-system-model}
We consider the simplest multi-user system that two co-channel users communicate with an AP. Let the complex fading coefficients between source 1 to AP be $h_1=\gamma_1 e^{j\alpha_1}$, and between source 2 to AP be $h_2=\gamma_2 e^{j\alpha_2}$, where $\gamma_1$ and $\gamma_2$ are denoted as the magnitude coefficients and $\alpha_1$ and $\alpha_2$ are denoted as phase coefficients of respective fading coefficients. Throughout this paper, we consider an equally strong links, where average link gains are the same, i.e., $\mathcal{E}\left\{|h_1|^2\right\}=\mathcal{E}\left\{|h_2|^2\right\}=g$, and it is assumed that AP has the full and exact knowledge of $h_1$ and $h_2$.\\ 
It is assumed that both users use scaler precoding, and let the unit power precoding coefficients be $e^{j\beta_1}$ and $e^{j\beta_2}$ for user 1 and 2 respectively. The received signal by the AP is:
\begin{align}\label{eq3}
r &= h_1 x_1 + h_2 x_2 + n = \gamma_1 e^{j\alpha_1} x_1 + \gamma_2 e^{j\alpha_2} x_2 + n,\\
&= \sqrt{E_s}\gamma_1 e^{j(\alpha_1+\beta_1)} s_1 + \sqrt{E_s}\gamma_2 e^{j(\alpha_2+\beta_2)} s_2 + n,
\end{align}
where $s_1,s_2 \in \mathcal{P}$. Herein $\mathcal{P}$ denotes $M$-PAM constellations, and $x_1$ and $x_2$ denote the transmitted signals of source 1 and 2 respectively. $n$ is the zero mean, variance, $\sigma^2$, additive white Gaussian noise (AWGN) sample. We also use following power normalizations as $\mathcal{E}\left\{|x_1|^2\right\}=\mathcal{E}\left\{|x_2|^2\right\}=E_s$ and $\mathcal{E}\left\{|s_1|^2\right\}=\mathcal{E}\left\{|s_2|^2\right\}=1$. It is AP's desire to detect both $s_1$ and $s_2$, and in the case of $M$-QAM, AP uses the following statistics to detect $s_1$ $\hat{y}_1=\frac{e^{-j(\alpha_1+\beta_1)}}{\sqrt{E_s}\gamma_1} r$, and AP may use the following maximum likelihood (ML) detector on $\hat{y}_1$ in order to detect $s_1$ as $\hat{s}_1=\min_{s_1 \in \mathcal{Q}} \left|\hat{y}_1 - s_1\right|$. However, in the case of BPSK of $M$-PAM, the following statistics is sufficient to detect $s_1$:
\begin{align}
\hat{y}_1 &= \text{Re}\left( \frac{e^{-j(\alpha_1+\beta_1)}}{\sqrt{E_s}\gamma_1} r \right), 
\end{align} 
Similarly, $\hat{y}_2=\text{Re}\left( \frac{e^{-j(\alpha_2+\beta_2)}}{\sqrt{E_s}\gamma_2} r \right)$ is sufficient to detect $s_2$. Without loss of generality, henceforth, the detection of data of the 1st user is considered exclusively, but the extension to detection of user 2 is straightforward. $\hat{y}_1$ can be simplified as:
\begin{align}\label{eq-1}
\hat{y}_1 &= s_1 + \frac{\gamma_2}{\gamma_1} \cos(\beta_2+\alpha_2-\beta_1-\alpha_1) s_2 + \frac{\tilde{n}_1}{\sqrt{E_s}\gamma_1}, 
\end{align}
where $\tilde{n}_1=n_I \cos\theta_1+n_Q \sin\theta_1$, $\theta_1=\alpha_1+\beta_1$, and $n_I$ and $n_Q$ are the real and imaginary parts of noise, $n_1$ respectively. Clearly, $\beta_2+\alpha_2-\alpha_1-\beta_1=k\pi/2$ for $k=1,3,5, \dots$, creates an interference free link for source 1, hence the following relationship for the precoding coefficients:
\begin{align}\label{eq0}
\beta_2 - \beta_1 = \alpha_1 - \alpha_2 \pm \pi/2,
\end{align}
where it is assumed that AP can derive $\alpha_1$ and $\alpha_2$. As a result, AP randomly selects a value which is also known by source 1 for $\beta_1$ and obtains $\beta_2$ from \eqref{eq0} which is subsequently fed-back to user 2. As a result, average BER of user 1 in $M=2$ (i.e., in BPSK or 2-PAM) becomes \cite{Gold05}:
\begin{align}\label{eq1}
\!\!\!P_{b}^1 &= \mathcal{E}\!\left\{\!Q\! \left(\sqrt{\frac{2E_s \gamma_1^2}{\sigma^2}}\right)\!\right\}\!=\mathcal{E}\!\left\{\!Q\! \left(\sqrt{\frac{2E_s|h_1|^2}{\sigma^2}}\right)\!\right\},
\end{align}
which is in fact equal to BER performance of interference free link, and $Q(.)$ in the right hand side of \eqref{eq1} is the standard $Q$-function. It is important note here that \eqref{eq0} also ensures interference free links for both user 1 and 2.
\subsection{Practical Considerations}\label{sec:practical-considerations}
Often exact feedback is not possible, and we hence consider a quantized feedback scheme. In this scheme user 2 requires the knowledge of $e^{j\beta_2}$ as accurately as possible. The AP uses a quantizer, $\mathcal{Q}_a$ based on an unit circle on the complex plan with $B$ equal size annular regions. Let the complex codebook of $\mathcal{Q}_a$ be $\mathcal{B}=\left\{c_1, \dots,c_k,\dots,c_B\right\}$. Consequently, the quantized version of $e^{j\beta_2}$, $e^{j\hat{\beta}_2}$ is mathematically given by:
\begin{align}
e^{j\hat{\beta}_2} &= \mathcal{Q}_a \left(e^{j\beta_2}\right) = \min_{c\in \mathcal{B}} \left|e^{j\beta_2} - c\right|^2,
\end{align}
and its codebook index is fed back to user 2 by using $\log_2 B$ number of bits.  
\section{Performance Analysis}\label{sec:performance-analysis}
Without loss of generality, the error performance of user 1 is considered. The expression in \eqref{eq1} gives the BER of 2-PAM with ideal CSIT. Understandably, the limited CSIT degrades BER performance, and Lemma \ref{lemma1} below summarizes BER of the proposed signaling scheme with limited CSIT.
\begin{Lemma}\label{lemma1}
	The average BER of user 1 with BPSK (2-PAM) of the signal processing scheme in Sec. \ref{sec:main-system-model} in Rayleigh fading with ideal CSI is available at AP, but with quantized CSIT being available at users is given by:
	\begin{align}\label{ber_eq:A}
	P_b^1 &= \mathcal{E}_z\left\{ Q\left(\sqrt{\frac{2E_s}{\sigma^2}}z\right) \right \}, 
	\end{align}
	where $z=\gamma_1+\gamma_2v$, and $-\infty \leq z \leq \infty$. Both $\gamma_1$ and $\gamma_2$ are Rayleigh distributed, and $v \sim \mathcal{U}[-\pi/B,\pi/B]$ in Rayleigh fading, where $B$ is the number of quantization levels for $\beta_2$.
\end{Lemma}
\begin{proof}
Let the quantized version of $\beta_2$, $\hat{\beta}_2$ is given by:
\begin{align} \label{eq12}
\hat{\beta}_2 &= \beta_2 + \epsilon_2 \stackrel{(a)}{=} \beta_1+\alpha_1 - \alpha_2 + \frac{\pi}{2} + \epsilon_2,
\end{align}
where $\epsilon_2 \in \left[-\pi/B, \pi/B \right]$ is the scalar quantization error, and $(a)$ is from \eqref{eq0}. Consequently, \eqref{eq-1} becomes:     
\begin{align}
\hat{y}_1 &= s_1 + \frac{\gamma_2}{\gamma_1}\cos \left(\frac{\pi}{2}+\epsilon_2 \right)s_2 + \frac{\tilde{n}_1}{\sqrt{E_s}\gamma_1}, \\
&= s_1 - \left( \frac{\gamma_2}{\gamma_1}\sin \epsilon_2\right) s_2 + \frac{\tilde{n}_1}{\sqrt{E_s}\gamma_1}, \\
&= s_1 - \frac{\gamma_2}{\gamma_1}v s_2 + \frac{\tilde{n}_1}{\sqrt{E_s}\gamma_1}
\end{align}
where $v=\sin \epsilon_2$ and $\tilde{n}_1=n_I \cos\theta_1+n_Q \sin\theta_1$.
The average BER of user 1 can be obtained by symmetry as:
\begin{align}
P_b^1 &= \text{Pr}\left(E | s_1=1, s_2=-1\right),
\end{align}
where $\text{Pr}\left(E | s_1=1, s_2=-1\right)$ denotes the average BER given $s_1=1$ and $s_2=-1$. This conditional probability can be simplified to give:
\begin{align}\label{eq10}
P_b^1 &= \text{Pr}\left(E | s_1=1, s_2=-1\right) = \mathcal{E}\left\{ \text{Pr}\left( \hat{y}_1 \leq 0 \right)_{\stackrel{s_1=1}{s_2=-1}} \right\}, \nonumber \\
&= \mathcal{E}\left\{ Q \left(\sqrt{\frac{2E_s}{\sigma^2}} \left(\gamma_1+\gamma_2v\right)\right) \right\}.
\end{align}
Note that the expectation in \eqref{eq10} is over $\gamma_1$, $\gamma_2$ and $v$, and the argument inside $Q$ function in \eqref{eq10} can also be negative. Let $z=\gamma_1+\gamma_2v$ for ease of notation, and hence the average BER of BPSK of the proposed scheme can be obtained as in Lemma \ref{lemma1}.
\end{proof}
The probability density function (PDF) of $v$ is approximately uniformly distributed for large $B$, and it is proved in Lemma \ref{Lemma2}.  
\begin{Lemma}\label{Lemma2}
	Let the sine of quantization error of $\beta_2$ be $v=\sin \epsilon_2$, and $v \sim \mathcal{U}[-\pi/B,\pi/B]$, where $B$ is the number of quantization levels.
\end{Lemma}
\begin{proof}
	From Rayleigh fading assumption for $h_1$ and $h_2$, we have $\alpha_1,\alpha_2 \sim \mathcal{U}[0,2\pi]$. Though in normal circumstances, $\alpha_1 -\alpha_2$ is triangular distributed, the quantity that matters is $\left( \alpha_1-\alpha_2\right) \mod{} 2\pi$, and it is uniformly distributed. Hence, the radian angle between, $e^{j\hat{\beta}_2}$ and $e^{j\beta_2}$ is uniformly distributed over support $\left[-\pi/B,\pi/B\right]$. The definition in \eqref{eq12} gives that the angle between $e^{j\hat{\beta}_2}$ and $e^{j\beta_2}$ is indeed $\epsilon_2$, and hence $\epsilon_2 \sim \mathcal{U}\left[-\pi/B,\pi/B\right]$. Consequently, $\sin \epsilon_2$ is Arcsine distributed. However, $\sin \epsilon_2 \approx \epsilon_2$ for small $\epsilon_2$, and hence, $\sin \epsilon_2 \simeq \mathcal{U}\left[-\pi/B,\pi/B\right]$.
\end{proof}
It is important to note here that BER expression in \eqref{ber_eq:A} is significantly different from the BER expressions of conventional systems. The argument inside $Q$-function in \eqref{ber_eq:A} is not strictly positive, and hence, average BER evaluation is significantly involving. In order to obtain the average BER in Lemma \ref{lemma1}, the PDF of $z$ needs to be evaluated.
\subsection{The Derivation of PDF of $z=\gamma_1+\gamma_2v$.}
From probability theory, CDF of $z$:
\begin{align}\label{eqA1}
F_z\left(t\right) &= \text{Pr} \left(\gamma_1+\gamma_2v \leq t\right) = \text{Pr} \left(\gamma_1\leq t -\gamma_2 v\right).
\end{align}
Slight abuse of notation, we use the same symbols for both RVs and their realizations. Unlike in conventional cases, $z$ can be negative. We have the following:
\begin{align} \label{eqA11}
F_z\left(t\right) &= \begin{cases}
F^+_z\left(t\right) & t \geq 0,\\
F^-_z\left(t\right) & t < 0.
\end{cases}
\end{align}
Let $\mathcal{A}$ be the region on $\gamma_2v$-plane such that $t \geq \gamma_2v$ and $t \geq 0$. Hence, $F^+_z\left(t\right)$ can be written as:
\begin{align}\label{eqA13}
F_z^+ \left(t\right) &=
\frac{B}{2\pi} \iint_{\mathcal{A}} \left(1-e^{\frac{(t-\gamma_2v)^2}{g}}\right) f_{\gamma_2}\left(y\right)dy dv,
\end{align}
and let $\mathcal{B}$ be the region on $\gamma_2v$-plane such that $t \geq \gamma_2v$ and $t < 0$. Hence, $F^-_z\left(t\right)$ can be written as:
\begin{align}\label{eqA14}
F_z^- \left(t\right) &=
\frac{B}{2\pi} \iint_{\mathcal{B}} \left(1-e^{\frac{(t-\gamma_2v)^2}{g}}\right) f_{\gamma_2}\left(y\right)dy dv,
\end{align} 
where the fact that PDF of $\gamma_1$, $f_{\gamma_1} (x) = \frac{2x}{g}e^{-\frac{x^2}{g}}$ is used, and the fact that $v$ is uniformly distributed over $[-\pi/B,\pi/B]$ is also already incorporated into \eqref{eqA13} and \eqref{eqA14}. Note that pdf of $\gamma_2$ is $f_{\gamma_2} (y) = \frac{2y}{g}e^{-\frac{y^2}{g}}$. The \eqref{eqA13} and \eqref{eqA14} can be elaborated as:
\begin{align} \label{eqA4}
F^+_z\left(t\right) &= \frac{B}{2\pi}\int_{-\frac{\pi}{B}}^{0} \int_{0}^{\infty} \left(1 - e^{\frac{(t-yv)^2}{g}}\right) f_{\gamma_2}\left(y\right)dy dv \nonumber \\
&+ \frac{B}{2\pi}\int_{0}^{\frac{\pi}{B}} \int_{0}^{\frac{t}{v}} \left(1 - e^{\frac{(t-yv)^2}{g}}\right) f_{\gamma_2}\left(y\right)dy dv.
\end{align}
\begin{align} \label{eqA6}
F^-_z\left(t\right) &= \frac{B}{2\pi}\int_{-\frac{\pi}{B}}^{0} \int_{\frac{t}{v}}^{\infty} \left(1 - e^{\frac{(t-yv)^2}{g}}\right) f_{\gamma_2}\left(y\right)dy dv.
\end{align}  
The \eqref{eqA4} can be simplified to obtain as in \eqref{eqA5}, where a dummy function is used for the reasons of space as:
\begin{align}
\psi\left(t,v\right) &= \sqrt{\frac{\pi}{g(1+v^2)}} \frac{tv}{1+v^2} e^{-\frac{t^2}{g(1+v^2)}},
\end{align}
and $\Phi \left(.\right)$ is the standard error function. This completes the derivations. Similar fashion, $F^-_z\left(t\right)$ can also be simplified, but the result is omitted for reasons of space. The average BER can then be evaluated as:
\begin{align} \label{eqA7}
P_b^1 &= \int_{-\infty}^{0} Q \left(\sqrt{\frac{2E_s}{\sigma^2}} t \right) f_z^-\left(t\right) dt \nonumber \\ 
& \qquad \qquad \qquad + \int_{0}^{\infty} Q \left(\sqrt{\frac{2E_s}{\sigma^2}} t \right) f_z^+\left(t\right) dt,
\end{align}
where $f_z^-\left(t\right) = \partial F_z^-\left(t\right)/\partial t$ and $f_z^+\left(t\right) = \partial F_z^+\left(t\right)/\partial t$. The \eqref{eqA7} will result a multiple integral, and the exact evaluation appears to be complex with little insight. Hence, semi analytical approach is used herein to evaluate \eqref{eqA7}. 
\begin{figure*}[t!]
	\begin{align} \label{eqA5}
	F^+_z\left(t\right) &= 1 - \frac{B}{2\pi} \int_{0}^{\frac{\pi}{B}} \frac{\left(2+v^2\right)e^{-\frac{t^2}{gv^2}}}{1+v^2} dv + \frac{B}{2\pi} \int_0^{-\frac{\pi}{B}} \psi \left(t,v\right) dv - \frac{B}{\pi} \int_0^{\frac{-\pi}{B}} \psi\left(t,v\right) \Phi \left(\frac{tv}{\sqrt{g\left(1+v^2\right)}}\right) dv \\ 
	& \qquad \qquad \qquad \qquad \qquad \qquad \qquad - \frac{B}{2 \pi} \int_0^{\frac{\pi}{B}} \psi\left(t,v\right) \Phi \left(\frac{t}{v\sqrt{g\left(1+v^2\right)}}\right) dv, \nonumber  
	\end{align}
\hrule	
\end{figure*}
\section{Numerical Results}
In this section, we compare the average bit error performance (BER) of the proposed signal processing scheme with exiting schemes, namely the optimum joint ML in \eqref{eq:ML:A} in Rayleigh fading. Let the bit-energy-to-noise ration is defined as $\text{EbNo}=E_sg/\sigma^2$. Fig. 1 shows the averaged BER performance (of one of the users) of the proposed scheme for BPSK with ideal CSIT, and as can be seen on Fig. 1, proposed scheme outperforms optimum joint ML in \eqref{eq:ML:A} by approximately 1.67 \text{dB}. Fig. 2 validates the accuracy of Lemma \ref{Lemma2}, and it can be seen that Lemma \ref{Lemma2} holds for $B$ values even as low as $B=8$, which is equivalent to $3$-bit feedback. Fig. 1 again shows the BER of the proposed scheme with limited CSIT for $4/6/8$-bit feedback. As shown in Fig. 1, 4-bit feedback can achieve ideal CSIT performance upto about 10dB, and so does 8-bit feedback upto 30dB. Similarly, the proposed scheme with 4-PAM outperforms joint ML in \eqref{eq:ML:A} with 4-QAM, and results are omitted for reasons of space. Similar results also
hold for multi-antenna uplink, downlink and point-to-point MIMO communication as well. For instance:    
\subsection{Multi-Antenna AP}\label{sec:multi-antenna-ap}
Let $n_r$ be the number of antennas at AP, and consequently in uplink, the received signal is given by:
\begin{align}
\bm{r} &= \bm{h}_1 x_1 + \bm{h}_2 x_2 + \bm{n}, \\
&= \sqrt{E_s}\bm{h}_1 e^{j\beta_1}s_1 + \sqrt{E_s} \bm{h}_2 e^{j\beta_2}s_2 + \bm{n}, 
\end{align} 
which is an extended version of \eqref{eq3}, where $\bm{r},\bm{h}_1,\bm{h}_2,\bm{n} \in \mathcal{C}^{n_r \times 1}$, and $\bm{h}_1,\bm{h}_2 \sim \mathcal{CN}\left(\bm{0}, g \bm{I}\right)$. Furthermore, AWGN vector is distributed as $\bm{n} \sim \mathcal{CN}\left(\bm{0}, \sigma^2 \bm{I}\right)$. An analysis similar to the one in Sec. \ref{sec:main-system-model}, gives that the condition, $\beta_2 - \beta_1 = \pm \pi/2 - \omega$, in conjunction with maximal ratio combining (MRC) ensures interference free communication for both user 1 and 2, and also ensures $n_r$-order diversity, where $\omega=\arg \left(\bm{h}_1^H\bm{h}_2\right)$. This is a sharp contrast to typical multi-user systems with
linear receivers \cite{Dush_general_jnl}. 
\begin{figure}[t]
	\centerline{\includegraphics*[scale=0.6]{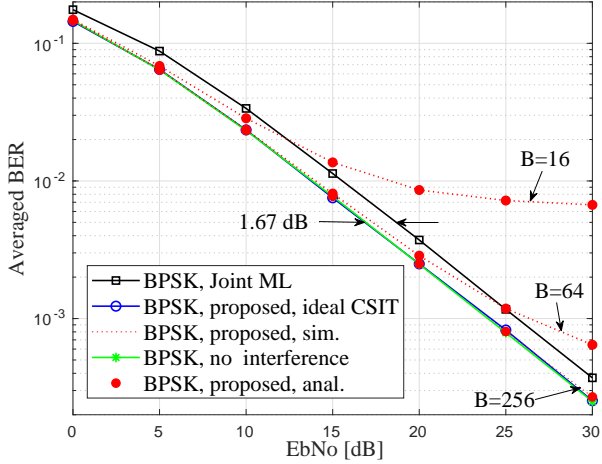}}
	\caption{BER performance of Joint-ML and the proposed system with difficult interference in Rayleigh fading.}
	\label{fig:con_en:fig1}
\end{figure}
\begin{figure}[t]
	\centerline{\includegraphics*[scale=0.6]{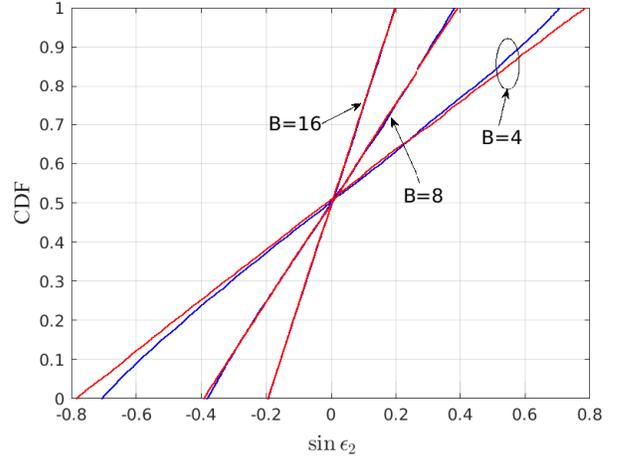}}
	\caption{Comparison of CDF of $\sin \epsilon_2$ in Lemma \ref{Lemma2} with uniform distribution (RED line) for different values of $B$.}
	\label{fig:con_en:fig2}
\end{figure}
\section{Conclusions}
This paper has proposed and analyzed a simple feedback scheme for multi-user communication in difficult interference. We compare the error performance of the proposed scheme with optimum joint ML detection. We presented an analysis for BER in fading along with key analytical challenges, and showed that proposed scheme outperforms even joint ML detection, which is widely considered as a fundamental limit, by 1.67dB in its simplest setting. The future work includes proposing approximations for the analytical results in Sec. 3, and a study to quantify the system-level implications of the proposed scheme. 
%
%
%
 
%
%
\end{document}